\documentclass[10pt, journal, compsoc]{IEEEtran}
\usepackage{amsmath}
\usepackage{bm, graphicx, bbm}
\usepackage{algorithm}
\usepackage{color}
\usepackage{epsfig}

\newcommand {\cD}{{\mathcal{D}}}

\newcommand {\cE}{{\mathcal{E}}}
\newcommand {\cG}{{\mathcal{G}}}

\newcommand {\cT}{{\mathcal{T}}}

\newcommand {\cV}{{\mathcal{V}}}

\newcommand {\bM} {{\mathbbm M}}

\newcommand {\bd} {{\bf d}}

\newcommand {\bo} {{\bf o}}

\newcommand {\bs} {{\bf s}}

\newcommand {\bw} {{\bf w}}
\newcommand {\bmu} {\boldsymbol{\mu}}

\newcommand{\avg}{{\rm avg}}

\newcommand {\Z} {{\mathbbm Z}}
\newcommand {\N} {{\rm I\kern-1.5pt N}}
\newcommand {\R} {{\rm I\kern-2.5pt R}}
\newcommand {\C} {{\rm I\kern-5pt C}}

\newtheorem{theorem}{Theorem}

\newcommand{\beqa}{\begin{eqnarray}}
\newcommand{\eeqa}{\end{eqnarray}}
\newcommand{\beqan}{\begin{eqnarray*}}
\newcommand{\eeqan}{\end{eqnarray*}}
\newcommand{\beq}{\begin{equation}}
\newcommand{\eeq}{\end{equation}}

\newcommand{\bfl}{\begin{flushleft}}
\newcommand{\efl}{\end{flushleft}}
\newcommand{\bgnv}{\begin{verbatim}}
\newcommand{\ednv}{\end{verbatim}}

\newcommand{\mydef}{& \hspace{-0.1in} := & \hspace{-0.1in}}
\newcommand{\myeq}{& \hspace{-0.1in} = & \hspace{-0.1in}}

\newcommand{\lb}{\nonumber \\}
\newcommand{\myarr}{\begin{array}{lll}}

\newcommand{\bitem}{\begin{itemize}}
\newcommand{\eitem}{\end{itemize}}
\newcommand{\benum}{\begin{enumerate}}
\newcommand{\eenum}{\end{enumerate}}

\newcommand{\E}[1]{{\mathbbm E}\left[ #1 \right]}
\newcommand{\bP}[1]{{\mathbbm P}\left[ #1 \right]}
\newcommand{\myhb}{\hspace{-0.3in}}
\newcommand{\myhf}{\hspace{0.3in}}

\newcommand{\myskip}{\\ \vspace{-0.1in}}

\newcommand{\ER}{Erd$\ddot{\rm o}$s-R$\acute{\rm e}$nyi }

\def\QED{~\rule[-1pt]{5pt}{5pt}\par\medskip}
\newenvironment{proof}{{\bf Proof: \ }}{ \hfill \QED}

\ifCLASSOPTIONcompsoc
  \usepackage[nocompress]{cite}
\else
  \usepackage{cite}
\fi

\begin{document}
%
\title{Influence of Clustering on Cascading 
	Failures in Interdependent Systems}


\author{Richard J. La\thanks{This work was 
supported in part by contract
70NANB16H024 from National Institute
of Standards and Technology.} 
\thanks{Author is with the Department of Electrical \& 
Computer Engineering (ECE) and the Institute for Systems 
Research (ISR) at the University of Maryland, College Park.
E-mail: hyongla@umd.edu}
}

\IEEEtitleabstractindextext{%
\begin{abstract}
We study the influence of clustering, more specifically 
triangles, on cascading failures in interdependent 
networks or systems, in which we model the dependence
between comprising systems using a dependence graph. 
First, we propose a new model that
captures how the presence of triangles in the dependence
graph alters the manner 
in which failures transmit from affected systems to others. 
Unlike existing models, the new model allows us to 
approximate the failure propagation dynamics using a 
multi-type branching process, even with triangles. 
Second, making use of the 
model, we provide a simple condition that indicates how 
increasing clustering will affect the likelihood that
a random failure triggers a cascade of failures, which
we call the probability
of cascading failures (PoCF). In particular, our condition
reveals an intriguing observation that the influence of 
clustering on PoCF depends on
the vulnerability of comprising systems to an increasing 
number of failed neighboring systems and the current 
PoCF, starting with different {\em types} of failed 
systems. Our numerical studies hint that increasing
clustering impedes cascading failures under
both (truncated) power law and Poisson degree
distributions. 
Furthermore, our finding suggests that, as the degree distribution 
becomes more concentrated around the mean degree 
with smaller variance, increasing 
clustering will have greater
impact on the PoCF. 
A numerical investigation of networks with
Poisson and power law degree distributions reflects 
this finding and demonstrates that increasing 
clustering reduces the PoCF much faster under Poisson 
degree distributions in comparison to power law degree 
distributions.

\end{abstract}

}


\maketitle

\section{Introduction}
	\label{sec:Introduction}

Many modern systems that provide critical 
services we rely on consist of 
interdependent systems.
Examples include modern information and 
communication networks/systems,  
power systems, manufacturing systems, and 
transportation systems, among others.
In order to deliver their services, the
comprising or component systems (CSes) 
must work together and oftentimes support 
each other.

Growing interdependence among CSes 
also exposes a source of vulnerability:  
the failure of one CS can spread to other 
CSes via dependence because a CS may no
longer be able to function without the 
support from failed CS(es). From 
this viewpoint, it is obvious that the overall
robustness of the system to failures 
will depend on the underlying dependence 
structure among CSes. 
We adopt a {\em dependence graph} to 
capture such interdependence among CSes, 
which we assume is {\em neutral}, i.e., 
there are no degree correlations between
neighbors. 

Despite a growing interest in modeling 
and understanding the robustness of large, 
complex systems
(e.g., \cite{Albert00, Baxter12, 
Buldyrev10, Kenett14, Rosato08, 
Shao11, Son12, Vesp, Zhuang16}),
intricate interdependence among CSes 
makes their analysis 
challenging. Furthermore, to the best of our
knowledge, there is no extensive theory 
or design guidelines that
allow us to answer even seemingly basic 
questions. 

Many of earlier studies that investigated
the cascading behavior of failures 
in interdependent
networks or systems, including our own, 
assumed tree-like propagation of failures 
(e.g., \cite{La_TON16, La_TON17, 
La_TNSE, Watts02, Yagan12}). But, it is well
documented that many real networks exhibit
much higher {\em clustering}, more formally 
known as {\em transitivity}, than classical
random graphs (e.g., \ER random graphs 
or configuration models
\cite{Bollobas, MolloyReed95}). Although how 
clustering is introduced in different real 
networks is still an open question, 
this observation led to new models that 
can generate clustered random graphs, e.g., 
\cite{JacobMorters15, Miller09, Newman09}.

Of particular concern is a widespread 
outbreak of failures among CSes, which 
can compromise the function of
the overall system, thereby risking
a potentially catastrophic system-level 
failure. We refer to {\em a cascade
of failures} or {\em cascading failures} 
as an event in which
the system experiences widespread failures 
well beyond the local neighborhood around
the initial failure and, in the absence of 
remedial actions, the spread of failures
slows down only when newly failed CSes no
longer have other CSes they can cause to fail. 
A key question we are interested in is: 
{\em how does clustering observed in 
real networks change the likelihood that
a random failure of a CS leads to a cascade
of failures in a large system?}

The goal of our study is two-fold: (i) to 
propose a new model that will allow researchers
to borrow existing tools
in order to study the 
effects of clustering or transitivity, and 
(ii) to complement  
existing studies (summarized in 
Section~\ref{sec:Related}) and contribute to the 
emerging theory on complex systems, by examining 
the impact of clustering
on the robustness of the systems with respect to 
random failures of CSes. Our hope is 
that the new findings and insight reported here
will help engineers and researchers better 
understand the influence of critical system 
properties, including clustering, 
and incorporate them into
design guidelines of complex systems. 

In order to answer the aforementioned question of 
interest to us, we develop a new 
model for capturing the influence of clustering
on the likelihood of a random failure triggering
cascading failures in large systems comprising
many CSes.
More precisely, it models the effect of
triangles in a dependence graph on transmission
of failures from affected CSes to their 
neighbors, as did the authors of \cite{Hackett11, 
Miller09}, but in a very different fashion.

Our model differs from existing models
employed to study similar effects: they account 
for the influence of triangles (or other short
cycles)
either by modifying degree distributions to
model the joint distributions of independent
edge degrees and triangle degrees
\cite{Hackett11, Miller09, Zhuang16}
or by replicating nodes to create cliques of size
equal to their degrees
\cite{CoupLelarge1, CoupLelarge2}. 
Instead, our model identifies the scenarios where 
a triangle among three CSes in the dependence 
graph alters the manner in which failures 
propagate locally.
In so doing, it allows us to capture the 
dynamics of failure propagation with 
the presence of triangles. 

A key benefit of our model is that 
it allows us to leverage an extensive set of 
tools available for branching
processes: 
as pointed out in \cite{Miller09}, it was believed
that the presence of short cycles due to high 
clustering renders the theory of 
(multi-type) branching processes
inapplicable. However, we will demonstrate that, by 
keeping track of what we call {\em immediate}
parents and children, we can approximate failure
propagation using a multi-type branching process. 
As a result, we are able to use existing tools to
estimate the likelihood of experiencing 
cascading failures in large systems. This in turn 
allows us to derive a simple, yet intuitive 
condition (Theorem~\ref{thm:1}) 
that tells us how increasing clustering 
would change the likelihood of experiencing
cascading failures. 

Our findings, to some extent, 
corroborate earlier findings obtained using 
different models. Moreover, our main
result (Theorem~\ref{thm:1}) reveals interesting
insight that sheds new light on the complicated
relation between system parameters and the
influence of clustering.  

Although our study is carried out in the framework 
of propagating failures in interdependent systems, 
we suspect that the basic 
model and approach as well as key findings are
applicable to other applications with suitable 
changes. These applications include (i) information 
or rumor propagation or new technology adoption via
social networks, (ii) an epidemic of disease in 
a society (e.g., cities or countries), 
and (iii) spread of malware in the Internet.

\subsection{A summary of main contributions}
	\label{subsec:Contributions}
	
The main contributions of our study can be 
summarized as follows.
\myskip

\noindent {\bf F1.} We propose a new model that 
captures the influence of triangles 
in dependence graphs
on failure spreading dynamics. Unlike
existing models (e.g., \cite{CoupLelarge1, 
CoupLelarge2, Miller09, Newman09}), 
we explicitly model the manner in which a
triangle alters how 
failure transmits between CSes. 
Therefore, it enables us to study the effect
of clustering in dependence graphs on 
failure propagation dynamics and the 
vulnerability of the system to cascading 
failures, which we measure using the probability 
of cascading failures (PoCF). 

As mentioned earlier, the model allows us to 
borrow a rich set of existing tools
by approximating failure
propagations using a multi-type branching 
process~\cite{Harris}. Moreover, it separates 
the influence of clustering on PoCFs from that 
of degree correlations, which is a problem  
observed with some existing models
\cite{Hackett11, Miller09}.  

\noindent {\bf F2.} Our main finding 
(Theorem~\ref{thm:1}) illustrates
that the influence of clustering on PoCF
is rather complicated in that it depends
on the ratio of PoCFs and infection 
probabilities of different {\em types}
of CSes, which will be defined precisely
in Section~\ref{sec:Types+IP}. 
However, there exists a simple condition 
that tells us whether
higher clustering facilitates or
impedes cascading failures. 

\noindent {\bf F3.} Numerical studies reveal
that increasing clustering tends to impede 
cascades of failures, rendering the 
system more robust to random failures. 
In addition, clustering has greater 
influence on PoCFs when the degree distribution 
in the dependence graph 
is more concentrated around the mean with smaller
variance. In particular, our 
study indicates that the PoCF decreases more 
rapidly with increasing clustering under Poisson 
degree distributions in comparison to power
law degree distributions. We offer an 
intuitive explanation
for this observation on the basis of our model
and main finding (Theorem~\ref{thm:1}). 
\myskip

A few words on notation: throughout the paper, we will
use boldface letters or symbols to denote (row) 
vectors or vector functions.\footnote{All vectors 
are assumed to be row vectors.}  For instance, $\bd$ 
denotes a vector, and the $j$-th element of $\bd$ 
is denoted by $d_j$. 
Vector ${\bf 1}$ represents the 
vector of ones of an appropriate dimension. 
The set $\Z_+$ (resp. $\N$) denotes the set of 
nonnegative integers $\{0, 1, 2, \ldots\}$ (resp. 
positive integers $\{1, 2, 3, \ldots\}$).  Finally, 
all vector inequalities are assumed componentwise.

The rest of the paper is organized as follows:
Section~\ref{sec:Related} summarizes existing 
studies that are most closely related to our
study. Section~\ref{sec:Model} delineates the
dependence graphs and infection graphs we 
use to model the propagation of failures
throughout the system, followed by a more
detailed description of the transmission of
failures among CSes and the influence of 
triangles in Section~\ref{sec:Types+IP}. 
Section~\ref{sec:MTBP} outlines the multi-type 
branching process we employ to approximate 
failure propagations and PoCFs in large 
systems. Sections~\ref{sec:Main} and
\ref{sec:Numerical} present
our main analytical finding and numerical 
studies, respectively. We provide the 
proof of the main finding in Section
\ref{appen:thm1} and then conclude in 
Section~\ref{sec:Conclusion}.

\section{Related Literature}
	\label{sec:Related}
	
There is already extensive literature on the topic
of epidemics, information propagation and 
cascading failures
(e.g., \cite{Albert00, Blume11, Boguna03b, 
Buldyrev10, MolloyReed95, Watts02}), which cuts across
multiple disciplines (e.g.,
epidemiology~\cite{CardyGrassberger85, Grassberger83, 
Pastor05, Schneider11}, 
finance~\cite{Caccioli12, Caccioli14}, 
social networks~\cite{Hu14, Moharrami}, and
technological networks~\cite{Cohen1, Cohen2,
La_TON16}). 
Given the large
volume of literature, it would be an unwise
exercise to attempt to provide a summary of
all related studies. Furthermore, although
the topic of clustering has seen a renewed 
interest recently, 
especially in the context of social networks and 
technological networks and its role in failure 
spreading, 
it has been studied in the past
in different settings, including population 
biology and epidemiology~\cite{Eames, Keeling99}. 
For these reasons, here we only 
summarize the most relevant studies that deal 
with the effects of clustering on cascading
behavior in networks. To improve readability,
even though there are some differences
in the employed models, we use the term 
`cascade' synonymously with `epidemic' and
`contagion' in the remainder of this section
because the studied events or phenomena
are similar.

As stated in Section~\ref{sec:Introduction}, 
earlier (classical) random graph
models proved to be unsuitable for describing 
many real networks; as the network
sizes increase, they typically lead to tree-like 
structures and fail to reproduce some salient
features of real networks, including the 
presence of short cycles which leads to higher
clustering. 
In order to address the shortcomings of the
random graph models, Miller~\cite{Miller09}
and Newman~\cite{Newman09} independently
proposed a new random graph model by extending
the classical configuration model, which can 
produce clustered networks with tunable
parameters: unlike in the classical  
configuration model where only a degree sequence
or probability is specified, the new model
specifies the joint probability $p_{st}(k_s, 
k_t)$ of two degrees 
-- independent edge degree $k_s$
and triangle degree $k_t$; an independent edge 
of a node is an edge with a neighbor which is 
not a part of a triangle. Since a triangle has 
two incident edges on the node, the total 
degree of a node with degrees $(k_s, k_t)$ is
$k_s + 2 \times k_t$.

Utilizing this new random graph model,
Miller~\cite{Miller09}
investigated the impact of clustering with 
respect to both cascade size and threshold. His
study shows that, although clustered networks
could exhibit a smaller cascade threshold, 
compared to a configuration model with an 
identical degree distribution, 
this is caused by the degree correlations, 
also known as assortativity or 
assortative mixing, 
introduced in the process of generating 
clustered networks. When networks with
similar degree correlations are compared, 
clustering
raises the cascade threshold and diminishes
the cascade size.  

In another study using the same model, Hackett
et. al~\cite{Hackett11} examined the influence
of clustering on the expected cascade size in 
both site and bond percolation as well as 
the Watts' random threshold model in $z$-regular
graphs. They showed that clustering reduces the
cascade size in the case of bond and site 
percolation. 
For the Watts' model in the $z$-regular graphs, 
their finding suggests that the impact of 
clustering depends on the value of $z$. 

In another line of interesting studies, 
Coupechoux and Lelarge proposed a new random graph
model where some nodes are replaced by cliques
of size equal to their degrees~\cite{CoupLelarge1, 
CoupLelarge2}. Making use of this new model, they 
examined the influence of clustering in social 
networks on diffusion and contagion in (random) 
networks. Their key findings include the observation 
that, in the contagion model with symmetric 
thresholds, the effects of clustering on cascade 
threshold depend on the mean node degree; for small 
mean degrees, clustering impedes cascades, whereas 
for large mean degrees, cascades are facilitated 
by clustering.

Although their contagion model is similar to our model, 
there is a key difference; while replicating the nodes 
to generate cliques in the random graphs, their model 
assumes that the contagion thresholds of all cloned 
nodes in a clique are identical. In our model, however, 
the nodes forming triangles can have different
thresholds, which we feel is more realistic in many
cases. As we will show, this seemingly innocuous 
difference has a significant impact on the effects of
clustering on cascading behavior. 

In~\cite{Zhuang16}, 
Zhuang and Ya$\breve{\rm g}$an extended the
model of Miller~\cite{Miller09} and
Newman~\cite{Newman09} to 
study information propagation in a multiplex
network with two layers representing 
an online social network (OSN) 
and a physical network, both
of which have high clustering. Only a subset of
vertices in the physical network are assumed to 
be active in the OSN.
Their key findings are: (a) clustering consistently 
impedes cascades of information to a large number of
nodes with respect to both the critical threshold
of information cascade and the mean size of
cascades; and (b) information transmissibility 
(i.e., average probability of information 
transmission over a link) has significant impact;
when the transmissibility is low, 
it is easier to trigger a cascade of information 
propagation with a smaller, densely connected 
OSN than with a large, loosely connected OSN. 
However, when the transmissibility is high, the 
opposite is true. 

In another study closely related to 
\cite{Zhuang16}, Zhuang et. al \cite{Zhuang17}
investigated the impact of the presence
of triangles on cascading behavior in multiplex 
networks, using a content-dependent linear 
threshold model. Their numerical studies show 
that the impact of increasing clustering depends 
on the mean degrees; when the mean degrees are
small, increasing clustering makes cascades
less likely. However, beyond some threshold
on the mean degrees, it has the opposite
effect. 

There are several major differences between 
these studies and ours. For instance,  
unlike these existing studies that primarily 
focused on (expected) cascade size or 
threshold, we focus on the likelihood that a 
random failure in the system will cause 
a cascade of failures in the system. 
Moreover, our model and main finding 
(Theorem~\ref{thm:1}) together offer some key
insights that are hard to obtain using the 
previous models and shed some light on how 
the degree distributions change the 
effects of clustering.

\section{System Model}
\label{sec:Model}

Let $\cV$ be the set of CSes. We model the 
interdependence among CSes using an undirected 
dependence graph $\cG = (\cV, \cE)$: the vertex set 
$\cV$ consists of the CSes, 
and undirected edges in $\cE$ between 
vertices indicate {\em mutual} 
dependence relations between the 
end vertices.\footnote{These dependence 
relations are not necessarily
the physical links in a network. For
example, in a power system, an overload failure
in one part of power grid can cause a failure
in another part that is not geographically close
or without direct physical connection to the former.} 
Two CSes with an undirected edge between them are 
said to be (dependence) neighbors. 

An undirected edge $e \in \cE$ in the dependence 
graph should be interpreted as a pair of directed 
edges pointing in the opposite directions. A directed
edge from a CS to another CS means that the latter 
can fail as a result of the failure of the former. 
In this sense, we say that the first CS {\em supports}
the latter CS. 
When we need to refer to directed edges in the 
dependence graph, we shall call it the 
directed dependence
graph (in order to distinguish it from the undirected
dependence graph or, simply, the dependence graph). 

For notational convenience, we often denote a generic
CS in the system by $a$ or $a'$ 
throughout the paper.

\subsection{Degree distributions
and clustering coefficient}

We denote the degree distribution (or probability 
mass function) of the dependence graph
by ${\bf p}_0 := (p_0(d); \ d \in \N)$, 
where $p_0(d)$ is
the fraction of CSes with $d$ neighbors. Thus, 
we assume that when we select a CS randomly, 
its degree can be modeled using a random variable 
(RV) with distribution ${\bf p}_0$.  

In our study, we are interested in modeling the 
effects of clustering in the dependence graph on 
the robustness of the system. We adopt the 
{\em clustering coefficient} in order to measure the 
level of clustering in the dependence graph
\cite{NewmanSIAM}: define a connected
triple to be a vertex with edges to two other
neighboring vertices. Then, the clustering 
coefficient is given by 
\beqa
C
\mydef \frac{\mbox{3 $\times$ \# of triangles 
	in the graph}}
	{\mbox{\# of connected triples in the graph}}.
	\label{eq:CC} 
\eeqa
In other words, the clustering coefficient 
tells us the fraction of connected 
triples that have an edge between the neighbors.
It is clear that this coefficient will be 
strictly positive if there exists at least 
one triangle in the graph.  
Throughout the study, we denote the clustering 
coefficient by $C \in (0, 1)$. 

As shown in \cite{NewmanSIAM}, the
clustering coefficient 
of real networks can be much larger than that
of random graphs. However, for most real networks, 
it does not exceed 0.2. In particular, all 
information and technological networks, which are
the networks of primary interest to us, have 
a clustering coefficient less than or equal to 
0.13. For this reason, we focus on scenarios 
where the clustering coefficient is not large
but not negligible.

\subsection{Propagation of failures}
	\label{subsec:Prop}

To study the robustness of a system to
failures, we need to model how a failure spreads 
from one CS to another. Here, we describe 
the model we employ to approximate the
dynamics of failure propagation between CSes.

We model failure propagation with the
help of a function 
$\wp_f: \N^2 \to [0, 1]$:
for fixed $d \in \N$ and $1 \leq n_f \leq d$, 
$\wp_f(d, n_f)$ tells us the probability that 
a CS with a degree $d$ will fail as well after 
$n_f$ of its $d$ neighbors 
fail.\footnote{This function 
was called the influence response function 
in \cite{Hackett11, WattsDodds07}.} Although 
it is not necessary, throughout the paper, 
we assume $\wp_f(d,d) = 1$ for all $d \in \N$, 
i.e., a CS fails when all of 
its supporting neighbors
fail, because it will be isolated from 
the rest of the system. 

An example that fits this model is the 
random threshold model used by Watts in
\cite{Watts02}, which is also used
in other studies (e.g., \cite{Blume11, 
Brummitt12}). In the Watts' model, 
every CS $a$ in $\cV$ has 
an intrinsic value $\xi_a$. These values
$\{ \xi_a; \ a \in \cV\}$ of CSes 
are modeled using mutually independent, 
continuous RVs 
with a common distribution $F$. We refer to
$\xi_a$ as the security 
state of CS $a$. 

In his model, a CS $a$ fails 
as a consequence of the failures
of its neighbors when the fraction of its 
failed neighbors exceeds its 
security state $\xi_a$. Therefore, 
for a given pair $(d, n_f)$ with $n_f \leq d$,  
$\wp_f(d, n_f)$ is equal to $\bP{\xi_a < 
n_f / d} = F(n_f / d)$.

\subsection{Infection graphs with triangles}

For our study, we focus on scenarios where each 
failed CS, on the average, affects only a small 
number of neighbors.\footnote{When each failed 
CS causes many other neighbors to fail as well, 
cascading failures are likely and should happen 
often. This may indicate that
the system is poorly designed. Instead,
we are interested in more realistic scenarios of 
interest in which cascading failures are possible
and do occur, but not too frequently.} Furthermore, 
we are primarily interested in scenarios in which, 
even when cascading failures occur, the fraction
of affected CSes is relatively small.
To make this more precise, we introduce {\em infection 
graphs}: starting with an initial failure 
of a CS in the system, the infection
graph consists 
of all failed CSes and the directed 
edges used to contribute to the failures of 
neighbors. In other words, a directed edge 
from CS $a$ to CS $a'$ in the directed 
dependence graph belongs to the 
infection graph if and only if CS $a$
fails before CS $a'$ does. If we replace the
directed edges in the infection graph with
undirected edges, we call it the undirected
infection graph. 

Unlike many earlier studies that assumed a
tree-like infection graph with no cycles, e.g., 
\cite{Brummitt12, La_TON16, La_TNSE,
 Watts02, Yagan12}, 
we allow for the existence of triangles in the
undirected infection graph. But, we assume
that the triangles are not very common in the
infection graph because the fraction
of affected CSes is small. For the same reason, 
we do not consider other larger cliques consisting of
more than three CSes in the infection graph, for 
such larger cliques would appear much less
frequently~\cite{Janson}.

\section{Agent Types and Neighbor Infection 
	Probabilities}
	\label{sec:Types+IP}

As mentioned in Section~\ref{sec:Introduction}, we
are interested in scenarios where the number of 
CSes in the system is large (at least in the 
order of thousands). As we 
will explain shortly, in a large system,  
the propagation of failures can be approximated
with the help of a multi-type branching process
under some simplifying assumptions. 
To be more precise, we shall borrow from the theory
of branching processes with finitely many types
in order to study the likelihood of a single 
initial failure leading to cascading
failures.

\subsection{Agent types}

In our model, depending on how a CS fails, 
there are three possible types we consider 
for the failed CS. More precisely, a CS
that experiences a random failure without any
failed neighbor is of type 0. For each $t \in 
\{1, 2\} =: \cT$, a failed CS is of type $t$ if
its failure was caused by those of $t$ 
failed neighbors. 
The CS(es) whose failures lead to the failure
of another CS, say $a$, are called the
{\em parent(s)} of $a$, and CS $a$ is 
called their {\em child}. Also, borrowing
from the language of epidemiology, we say that
the parent(s) {\em infected} the child. 

Obviously, we implicitly assume that every 
failed CS has no more than two parents. 
In general, it is possible that some high
degree CSes experience a failure
after more than two of their neighbors
fail. But, this would be uncommon
when the fraction of failed CSes is 
small, which is the scenario we focus on 
in this study. 

We shall discuss the distribution of the number of 
children of two different types which are infected
by a failed CS of type $t \in \{0, 1, 2\}$ shortly. 
To explain
these children (vector) distributions, we first 
need to describe how we approximate the probability 
with which a neighbor of failed CS(es) becomes
infected.

\subsection{Triangles in infection graphs}
	\label{subsec:TriangleIG}

In order to motivate the model we employ to 
approximate the propagation of failures, 
we begin with an 
illustrating example shown
in Fig.~\ref{fig:type}. In the figure, dotted
lines indicate undirected edges in the 
dependence graph, and solid directed arrows
represent the contribution to the failures
of children by failed CSes, i.e., parents.

\begin{figure}[h]
\centerline{
	\includegraphics[width=3.3in]{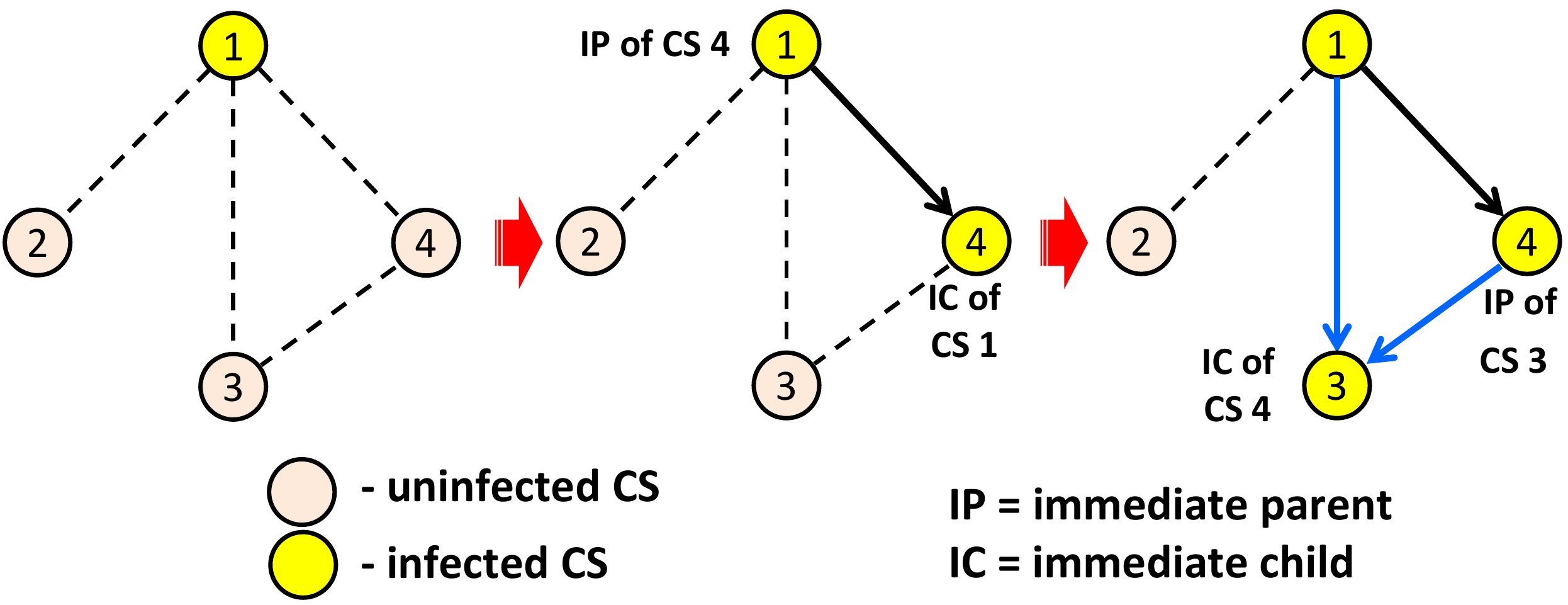}
}
\caption{Example of type 1 and 2 CSes.}
\label{fig:type}
\end{figure}

Initially, CS 1 is the 
lone failed CS and, as a result, CSes
2 through 4 face the possibility of
infection by CS 1. 
Suppose that CS 4 becomes infected 
while CSes 2 and 3 remain unaffected.
Since CS 4 has a single parent, namely
CS 1, it is an example of type 1 CS.
In this case, 
we call CS 4 an {\em immediate
child} of CS 1, and refer to CS 1 as
the {\em immediate parent} of CS 4. 

On the other hand, even though CS 3 first 
survives the failure of CS 1, once CS 4 
fails following its infection by CS 1, the
failures of CSes 1 and 4, both of which
are neighbors of CS 3, lead to that
of CS 3. In this case, because it 
takes the failures 
of two neighbors to cause that of CS 3, 
it is an example of type 2 CS. 
We call CS 1 and CS 4 the {\em first parent}
and the {\em immediate parent}, respectively, 
of CS 3. Analogously, we call CS 3 an 
immediate child of CS 4 (but not 
of CS 1). It will be clear why we call 
CS 1 and CS 4 the immediate parent of CS 
4 and CS 3, respectively, when we describe
the multi-type branching process to 
approximate failure propagation in 
Section~\ref{sec:MTBP}. 

Before we proceed, in order to explain  
the role of triangles in failure propagation
dynamics, 
we first describe several scenarios where
the presence of triangles in the dependence
graph does not affect the failure 
propagation dynamics. This will help us to 
isolate the cases of interest to us, in which 
the existence of triangles alters the 
way failures propagate. 

{\bf Sc1.} In the example shown in 
Fig.~\ref{fig:type}, we assumed that CS 1 is 
a parent of CS 4. Suppose instead that 
CSes 1 and 4 are not neighbors, but there are
additional failed CS(es) between them forming a 
path from CS 1 to CS 4 in the infection graph. 
In this case, since CSes 1 and 4 are 
not neighbors, CSes 1, 3 and 4 do not form a 
triangle. Thus, whether or not we model triangles 
in the dependence graph would not affect the 
transmission of failure to CS 3.

{\bf Sc2.} Even though in Fig.
\ref{fig:type} we assumed that CS 4 was
infected by CS 1 alone, it is also possible
that CS 4 has additional 
parent(s) that contributed to its failure.
From the viewpoint of modeling the influence 
of the triangle (amongst CSes 1, 3, and 4) 
on potential infection of CS 3 by CSes 1 and 
4, however, whether or not CS 4 has 
additional parent(s) is unimportant.

{\bf Sc3.} According to the way we
defined our infection graphs, three failed CSes 
that form a triangle in the dependence
graph may not form a triangle in the 
undirected infection graph. The reason for this
is that, when two neighboring CSes fail 
nearly simultaneously in such a way that the
failure of one does not contribute to that of
the other, there is no edge in the infection 
graph. Roughly speaking, there are two different
ways in which this can happen. 
\benum
\item[C1.] One CS fails first and infects the 
other two neighbors that fail almost
simultaneously, so that the latter two do not 
contribute to the failure of each other. 

\item[C2.] Two of the CSes first fail nearly
at the same time in a way that the failure of one
does not contribute to that of the other. These
two failed CSes then cause the third CS to fail. 
\eenum

In case C1, the existence of the triangle in the 
dependence graph among the three CSes does not 
play any role in the propagation of failures
and the infection of the latter 
two CSes would behave the same way even if 
they were not neighbors with each other. 
Similarly, in case C2, 
the existence of a triangle, more specifically
that of an edge between the first two failing
CSes, is irrelevant to the infection of 
the third CS. 
\myskip

These observations suggest that, from the viewpoint 
of modeling the effects of triangles in failure 
propagation, the only case in which the presence
of a triangle among three CSes matters for
transmitting failures is when the three CSes fail
one after another and each failed CS contributes
to ensuing failures of other CS(es).  
In other words, the first CS to fail 
contributes to the failure of the second
CS in the triangle, and the failures
of the first two CSes subsequently 
cause that of the third CS. This 
is precisely what our model captures 
as explained in the subsequent section.

Recall from scenario {\bf Sc1} that two
failed neighbors of an uninfected CS
might not be neighbors with each other. 
Therefore, in general 
the parents of a type 2 CS
need not be neighbors. When this is true,
however, the three CSes do not form a 
triangle and whether or not we model triangles
would not have any significant impact on our
study. Furthermore, if only a 
small fraction of CSes become infected
as we assumed earlier,
for most CSes with small to moderate 
degrees, the likelihood that they will
have two or more failed neighbors
that are not neighbors with each other 
would be small.
For these reasons, we do not 
consider or model the scenarios where
the two parents of a type 2 CS are
not neighbors. Put differently, we 
only model type 2 CSes whose parents
are neighbors and the three CSes
form a triangle as illustrated in 
Fig.~\ref{fig:type}, in order to examine
the effects of triangles in the
dependence graph. 

However, it should be evident that our
model can be extended to consider other
features, including larger 
cycles, by introducing additional
types of CSes necessary to capture
their presence. Obviously, this is likely
to increase its complexity and degrade
its tractability.

\subsection{Infection probability of neighbors}
	\label{subsec:InfProb}
	
For $t \in \cT$, we denote by $q_t$
the probability that a CS (without knowing
its degree) will fail following 
the $t$-th failure among the neighbors
and become a type $t$ CS. We shall refer
to $q_t$ as the {\em infection probability}
of type $t$ CS. Keep in mind that these
infection probabilities $q_t$, $t \in \cT$, 
do not depend on clustering coefficient
as we explain below. 
 
$\bullet$ {\bf Computation of $q_1$: } 
The infection probability $q_1$ is the 
probability with which CS 4 (or CS 2 or 3) 
is infected by CS 1 as the sole failed neighbor 
in Fig.~\ref{fig:type} without knowing the
degree of CS 4. 
According to the model outlined in 
Section~\ref{subsec:Prop}, a CS of degree $d$
with only a single failed neighbor will fail 
with probability $\wp_f(d, 1)$. Therefore, by 
conditioning on the degree of the neighbor, 
we can obtain the infection probability $q_1$ 
as
\beqa
q_1
\mydef \sum_{d \in \N} w_1(d) \ \wp_f(d, 1),  
	\label{eq:q1}
\eeqa 
where 
\beqa
w_1(d) := \frac{d \cdot p_0(d)}{d_{\avg}}
	\ \mbox{ for all } \ d \in \N,
	\label{eq:w1}
\eeqa
and
$d_{\avg} := \sum_{d \in \N} d \cdot p_0(d)$ is the
average degree of CSes. 

The reason that the degree distribution of a neighbor
used in (\ref{eq:q1}) is given by $\bw_1 := (w_1(d); 
\ d \in \N)$, as opposed to ${\bf p}_0$, was first
discussed in \cite{Callaway}: the degree 
distribution of the neighbor is the 
conditional degree distribution given that it is 
a neighbor of the failed CS
attempting to infect it. The probability 
that a neighbor of the failed CS has degree $d$ is 
proportional to 
$d$ and, hence, is approximately $w_1(d)$ 
under the assumption that the dependence graph
is neutral.

$\bullet$ {\bf Computation of $q_2$: }
The infection probability $q_2$ can be computed
in an analogous manner. Since $q_2$ is the
probability that a CS fails following
the second failure among neighbors and we assume
that the two failed neighbors form a triangle with 
the CS (from the previous subsection), this 
is the probability with which CS 3 fails in the
example of Fig.~\ref{fig:type}, following the 
failure of CS 4 after having survived the failure of
CS 1 (without knowing the degree of CS 3). 
First, since the CS has two failed neighbors, 
clearly its degree is at 
least two. Second, it must have survived the first 
failure of a neighbor before succumbing to the 
second failure among neighbors.

Because the degree of such a CS cannot be one, the 
conditional degree distribution is given by 
${\bf w}_2 := (w_2(d); \ d \in \N)$, where 
\beqa
w_2(1) = 0 \ \mbox{ and } \  
w_2(d) = \frac{w_1(d)}{1 - w_1(1)} \ 
	\mbox{ for } \ d \in \N \setminus \{1\}. 
	\label{eq:w2}
\eeqa 
Taking these
observations into account and conditioning
on the degree of CS, we obtain
\beqa
q_2
\mydef \sum_{d \in \N} w_2(d) (1 - \wp_f(d,1)) 
	\tilde{\wp}_f(d,2),
	\label{eq:q2}
\eeqa
where $\tilde{\wp}_f(d,2)$ is the probability that 
a CS of degree $d$ will fail after 
the second failure among neighbors, conditional on 
that it survived the first failure of a neighbor. 
Using the definition of the conditional probability, 
we get
\beqa
\tilde{\wp}_f(d, 2)
\mydef \frac{\wp_f(d, 2) - \wp_f(d, 1)}{1 - \wp_f(d,1)},
	\ d \in \N \setminus \{1\}. 
	\label{eq:twp}
\eeqa
Substituting (\ref{eq:w2}) and (\ref{eq:twp}) in 
(\ref{eq:q2}), we get
\beqa
q_2 
\myeq \frac{1}{1-w_1(1)} \sum_{d=2}^{\infty} 
	w_1(d) \big( \wp_f(d,2) - \wp_f(d,1) \big). 
	\label{eq:q2-2}
\eeqa

Let us explain the conditional probability 
$\tilde{\wp}_f(d, 2)$ given in (\ref{eq:twp}) using
the Watts' model. As explained earlier, a CS
fails when the fraction of failed neighbors exceeds
its security state $\xi$ and $\wp_f(d,n_f)
= F(n_f/d)$. Therefore, given that a CS of degree
$d$ survived the failure of a single neighbor, the 
conditional probability that it will fail when 
a second neighbor fails is equal to 
\beqan
\tilde{\wp}_f(d,2)
\myeq  \bP{\xi < \frac{2}{d} \ \Big| 
	\ \xi \geq \frac{1}{d}}
= \frac{ \bP{ \frac{1}{d} \leq \xi < \frac{2}{d}} }
	{ \bP{ \frac{1}{d} \leq \xi } } \lb 
\myeq \frac{ \wp_f(d, 2) - \wp_f(d, 1) }{1 - \wp_f(d, 1)}. 
\eeqan

From the degree distributions defined 
in (\ref{eq:w1}) and (\ref{eq:w2}) and 
used to compute the infection probabilities 
in (\ref{eq:q1}) and (\ref{eq:q2}), 
respectively, it is clear
that our model implicitly assumes a neutral
dependence graph and does not model 
any potential degree
correlations between CSes. It has been observed
(e.g., \cite{Hackett11, Miller09})
that one of challenges to studying the effects
of clustering is that, many existing models
for generating clustered networks unintentionally
introduce degree correlations or assortativity.
From the construction of cliques via replication 
of nodes with identical degrees, it is obvious
that similar degree correlations are introduced 
in the model proposed in \cite{CoupLelarge1, 
CoupLelarge2}. 

In~\cite{Miller09}, Miller points out that  
the influence of such degree correlations 
caused some researchers
to incorrectly conclude that clustering 
facilitates cascades of failures with respect
to cascade sizes and cascade threshold; he
demonstrates that, when compared against 
unclustered networks with similar degree 
correlations,
clustered networks have smaller cascade sizes
and higher cascade threshold. In other words, 
clustering has the opposite effects and 
impedes cascades of failures.

Our model does away with this problem by 
modeling the influence of clustering on 
transmission of failures, rather than 
trying to generate random graphs (of finite
size) explicitly. For this reason, any 
qualitative changes in the robustness of
systems we identify using the model can 
be attributed to 
clustering, without having to worry about the
influence of degree correlations.

\section{Multi-type Branching Process for Modeling
	the Spread of Failures} 
	\label{sec:MTBP}
	
We approximate the propagation of failures 
among CSes using a multi-type branching 
process (MTBP), in which
the members of the MTBP are of the types
defined in the previous section. 
We estimate the PoCF of a large system,
using the probability that the population 
in the MTBP does not die out. 

This approach is consistent with our notion of 
cascading failures provided in Section
\ref{sec:Introduction}: when cascading failures 
occur in a large system, they will likely 
spread to different parts of
the system well beyond the local neighborhood
near the initial random failure and 
slow down only when
newly failed CSes do not have neighbors
to infect any more. Thus, in the case of
an infinite network where the number of
CSes in the system is unbounded 
(as assumed in many studies, e.g., 
\cite{Hackett11, La_TNSE, Miller09, 
Zhuang17}), the spread of failures
will not cease when cascading
failures happen, which is equivalent
to the population not dying out in our MTBP. 

We assume that
initially a randomly chosen CS experiences
a failure. Since this CS does not have any 
parent, it is of type 0. Subsequently, 
when a type $t$ CS fails, it produces 
immediate children of two different types 
in accordance with the (immediate)
children distribution $h_t$ described below, 
independently of other infected CSes in 
the system.

\subsection{Multi-type branching processes}
	\label{subsec:MTBP}
	
The MTBP we use to approximate
failure propagation can be
constructed as follows: it starts with 
a single type 0 agent at the beginning
(the zero-th generation). Then, each member in 
the $k$-th generation ($k \in \Z_+$) produces 
immediate children of two different types, 
who become members of the $(k+1)$-th generation,
according to children distributions derived in 
Section~\ref{subsec:ChildDistribution}, 
independently of each other. 

Since an infected CS can have up to two 
parents and, when a CS has two parents, 
one of them is the parent of 
the other as shown in Fig.~\ref{fig:type}, 
the parents themselves 
can belong to two different generations. 
A failed CS in our infection graph belongs
to the $k$-th generation in the MTBP 
if and only if the length of the {\em longest} 
path in the infection graph from the root 
to the CS is equal to $k$.

\begin{figure}[h]
\centerline{
	\includegraphics[width=3.1in]{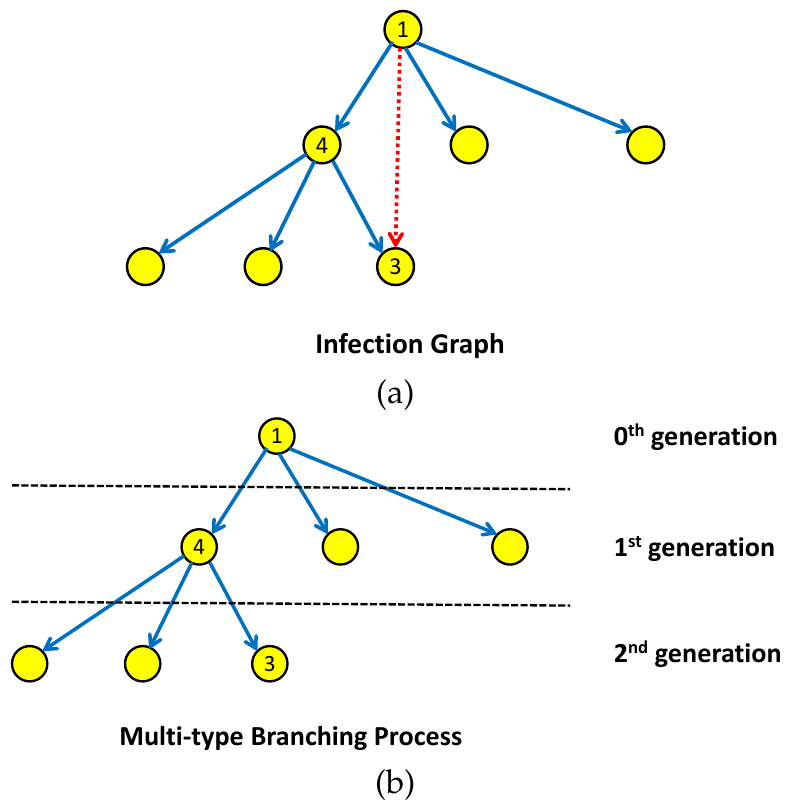}
}
\caption{(a) Infection graph, (b) corresponding 
	Multi-type branching process.}
\label{fig:MTBP}
\end{figure}

Let us illustrate this using the example
in Fig.~\ref{fig:type} with CS 1 as the
root of an infection graph, which 
experiences the first failure. The 
resulting infection graph is shown in 
Fig.~\ref{fig:MTBP}(a), where 
unnumbered circles
are other infected children that are 
not shown in Fig.~\ref{fig:type}.
In the figure, a solid blue arrow
indicates the influence of the 
immediate parent that causes the 
child to fail, whereas
a dotted red arrow represents the
influence of the first parent of
a type 2 CS which contributes to the 
failure.

In order to construct the corresponding
MTBP, we only keep the solid blue 
arrows and delete the dotted red arrows
from the infection graph. In other words, 
we only retain the arrows from 
immediate parents. Since every 
CS has a single (immediate) parent 
in the new graph, it is clear that this 
gives rise to a directed tree. 
This is shown in Fig.~\ref{fig:MTBP}(b). 
A failed CS is now a member of the $k$-th 
generation in MTBP if its distance from 
the root in this new tree is $k$.

\subsection{Immediate children distributions in 
	MTBPs}
	\label{subsec:ChildDistribution}

In order to complete the description of the 
MTBP that we employ to approximate failure
propagation, we need to specify the children
distribution for each type. Clearly, 
except for the initial
CS that experiences a random failure,
which is the root of the infection graph, 
all other infected CSes are of type 1 or 2, 
because they must have a parent. 
Thus, all children infected by a failed CS
would be of either type 1 or 2. 

$\bullet$ {\bf Neighbors of a type 0 CS:} 
Assume that CS $a$ experiences the initial
random failure. Since CS $a$ is the only 
failed CS at the moment, all of its neighbors 
face the risk of infection from a single 
neighbor. Thus, each neighbor will be 
infected by CS $a$ and become
its immediate child with 
probability $q_1 =: v^C_{0,1}$ or remain
unaffected with probability $1 - q_1 =: 
v^C_{0,0}$.

Because of the presence of triangles
in the dependence graph,
it is possible that some of immediate 
children of CS $a$ are also neighbors
of other unaffected neighbors of CS $a$. 
Suppose that one of these unaffected 
neighbors of CS $a$ fails as a result 
of the failure of a second neighbor, say 
CS $a'$, which was infected by CS
$a$. Then, it will be an immediate child of 
CS $a'$ and will not count as an immediate 
child of CS $a$, as explained in Section
\ref{subsec:TriangleIG}. 
\myskip

$\bullet$ {\bf Neighbors of a type 1 CS:} 
Suppose that a CS $a$ is infected by a 
single parent $a'$ and, hence, is of type 1. 
Recall from the definition in (\ref{eq:CC})
that the clustering coefficient is 
the likelihood that there is an edge between
the two neighbors of a connected triple. 
Thus, when the clustering coefficient is 
$C \in (0, 1)$, 
the probability that there is an edge between 
a randomly chosen neighbor of $a$ and its parent 
$a'$ can be approximated using $C$ because 
CS $a$ is a connected triple. 
Moreover, we assume that 
each neighbor of $a$ has an edge with the
parent $a'$ with probability $C$, independently
of each other, which is reasonable for
large systems.

Based on this assumption, we approximate the 
probability that a neighbor becomes another 
type 1 CS using the product $(1 - C) q_1
=: v^C_{1,1}$ of (a) the probability that it is 
not a neighbor of $a'$ ($1 - C$) and (b)
that of becoming infected by the failure
of CS $a$ ($q_1$). 
Similarly, the probability that the neighbor 
has an edge with $a'$ and 
becomes a type 2 CS is estimated using 
$C \cdot q_2 =: v^C_{1,2}$. 
Note that, when the
neighbor has an edge with the parent $a'$, in 
order for it to be a type 2 CS, it must have 
initially survived
the failure of $a'$ and then fail
following the failure of CS $a$. Hence, the 
probability that it will fail following
the infection of CS $a$ and become a type 2
CS would be $q_2$, 
not $q_1$, as explained in Section
\ref{subsec:InfProb}. 
Finally, the probability that it 
remains unaffected is $1 - (1 - C) q_1
- C \cdot q_2 =: v^C_{1,0}$. 

$\bullet$ {\bf Neighbors of a type 2 CS:} 
When CS $a$ is of type 2 with two parents, 
it is possible for a neighbor of CS $a$ to 
have an edge with both of the parents and, thus, 
have three failed neighbors. However, when the 
fraction of failed CSes is low or the 
clustering coefficient $C$ is 
not large, this would
not occur often. Furthermore, modeling CSes
with more than two infected neighbors would
require introducing additional types of CSes. 
For these reasons, we do not explicitly
model the 
cases where a CS faces more than 
two failed neighbors and approximate them 
using the case where the CS
has only two failed neighbors instead.

Making use of this argument, we assume that 
neighbors of a type 2 CS  
have an edge with (a) neither parent with 
probability $(1-C)^2$ and (b) one parent 
with probability $1 - (1-C)^2 = (2C - C^2)$, 
independently of each other. Note that
the possibility of being a neighbor with
both parents is folded into the latter
scenario of having an edge with only one
parent.

On the basis of this assumption, we estimate 
the probability that a neighbor of CS $a$
becomes a type 1 CS using $(1-C)^2 q_1 
=: v^C_{2,1}$ and that of
becoming a type 2 CS by $(2 C - C^2) q_2 
=: v^C_{2,2}$. Finally, the probability 
that the neighbor avoids infection is
equal to $1 - (1-C)^2 q_1 - (2 C - C^2) q_2
=: v^C_{2,0}$. When it is convenient and
there is no risk of confusion, 
we do not explicitly denote the dependence 
on clustering coefficient $C$ and write
$v_{t,t'}$ in place of $v^C_{t,t'}$. 
\myskip

\subsubsection{Conditional degree distributions
	of failed CSes}	
	
Given these probabilities, once we fix
the degree $d$ and type $t$ of a failed
CS, we can approximate the distribution
of the number of immediate children of 
different types,
which are produced by the CS with the 
help of multinomial distributions: 
the probability that it produces $o_1$
type 1 children and $o_2$ type 2 children
(with $o_1 + o_2 \leq d - t$) is given by 
\beqan
\left[ \frac{d-t}{o_1 : o_2} \right] 
v_{t,1}^{o_1} \ v_{t,2}^{o_2} 
\ v_{t,0}^{d-t-o_1-o_2}, 
\eeqan 
where $v_{0,2} = 0$, and 
\beqan
\left[ \frac{d-t}{o_1 : o_2} \right]
:= \frac{(d-t)!}{o_1! o_2! (d-t-o_1-o_2)!}, 
\eeqan
denotes a multinomial coefficient. 
Hence, we can estimate the children 
distribution by conditioning on the degree
of the CS.

The conditional degree distributions of 
infected CSes in general differ
from the prior distribution ${\bf p}_0$. We 
denote the conditional degree distribution of 
type $t$ CSes by ${\bf p}_t
:= (p_t(d); \ d \in \N)$, $t \in \cT$. 
Making use of the degree 
distributions of neighbors defined in 
(\ref{eq:w1}) and (\ref{eq:w2}) and the failure
probability function $\wp_f$, we can 
obtain the following conditional 
degree distributions of infected CSes: 
\beqan
p_1(d)
\myeq \frac{w_1(d) \ \wp_f(d, 1)}{q_1} 
	\ \mbox{ for } d \in \N,
	\label{eq:p1} \\  
p_2(d) 
\myeq \frac{w_1(d) (\wp_f(d,2) - \wp_f(d,1))}
	{(1-w_1(1)) q_2}
	\ \mbox{ for } d \in \N \setminus \{1\},   
	\label{eq:p2}
\eeqan
and $p_2(1) = 0$.

\subsubsection{Children distributions} 

Putting all the pieces together, 
we acquire the following (immediate) 
children distributions 
for different types: fix clustering 
coefficient $C \in (0, 1)$. The probability 
that a CS of type $t \in \{0, 1, 2\}$ 
will produce immediate children given by a 
children vector ${\bf o} = (o_1, o_2) \in 
\Z_+^{2}$, where $o_j$ is the number of 
type $j$ immediate children, is given by 
\beqa
&& \myhb h_t(\bo; C) \lb 
\myeq \sum_{\tau = o_1 + o_2}^\infty 
	\Big( p_t(\tau+t)
		\ \left[ \frac{\tau}{o_1:o_2} \right] 
	 v_{t,1}^{o_1} \
		v_{t,2}^{o_2} \
			v_{t,0}^{\tau - o_1 - o_2} \Big).  
		\label{eq:ht}
\eeqa

\subsection{Probability of extinction}
	\label{subsec:PoE}
	
Let $N_t(k)$, $t \in \cT$ and $k \in \N$, be 
the number of type $t$ CSes in the $k$-th 
generation of the MTBP. For instance, 
in Fig.~\ref{fig:MTBP}(b), $N_1(1) = 3$, 
$N_2(1) = 0$, $N_1(2) = 2$, and $N_2(2) = 1$. 
The probability $\bP{ \limsup_{k \to \infty} 
\ (N_1(k) + N_2(k)) = 0}$ is 
called the probability of extinction (PoE)
\cite{Harris}. Obviously, the PoCF 
is equal to one minus the PoE.

Let $\boldsymbol{\mu} = (\mu_1, \mu_2)$, 
where $\mu_t$ is the PoE, starting 
with a single type $t$ infected CS (instead
of a type 0 CS).
The probability of interest to us is the
PoE starting with a single type 0 CS,
which we denote by $\mu_0$. This can be 
computed from $\mu_1$ 
by conditioning on the degree of the 
CS, say $a$, which experiences the initial 
random failure. In other words, 
\beqa
\mu_0
\myeq \sum_{d \in \N} p_0(d) 
	\big( 1 - q_1 (1 - \mu_1) \big)^d.
	\label{eq:mu0}
\eeqa
Note that $q_1 (1 - \mu_1)$ is the probability
that a neighbor of CS $a$ will be infected
and then trigger a cascade of failures. 
Thus, $(1 - q_1(1-\mu_1))^d$ is the 
probability that none of the $d$ neighbors
gives rise to cascading failures. 
A key question of interest to us is how clustering
coefficient $C$ affects the PoE $\mu_0$ defined in
(\ref{eq:mu0}), which is a strictly increasing
function of $\mu_1$.

For each $t \in \cT$, let
$\E{h_t} = (\E{h_{t,t'}}; \ t' \in \cT)$
be a $1 \times 2$ row vector,
whose $t'$-th element is the expected number of
type $t'$ immediate children from a type 
$t$ CS. Define $\bM = [M_{t,t'}]$ to be 
a $2 \times 2$ matrix, whose
$t$-th row is $\E{h_t}$, i.e., 
$M_{t,t'} = \E{h_{t,t'}}$ for all $t, t'
\in \cT$. Let $\rho(\bM)$ denote the spectral
radius of $\bM$~\cite{matrix}.

It is well known \cite{Harris} that 
$\boldsymbol{\mu} = {\bf 1}$ if (i) $\rho(\bM)
< 1$ or (ii) $\rho(\bM) = 1$ and there is
at least one type for which the probability 
that it produces exactly one child is not 
equal to one. Similarly, 
if $\rho(\bM) > 1$, then $\boldsymbol{\mu}
< {\bf 1}$ and there is strictly positive
probability that spreading failures continue 
forever in an infinite system, suggesting
that there could be a cascade of failures
in a large system.

\section{Main Analytical Result}
	\label{sec:Main}

As stated in Section~\ref{sec:Introduction}, 
our goal is to understand how
clustering among CSes in the dependence
graph affects the robustness of the system.
As it will be clear, the relation between 
the PoCF (or equivalently, PoE) and system
parameters, including clustering coefficient,
is rather complicated. In particular, 
our main finding stated in Theorem~\ref{thm:1}
below suggests that whether increasing 
clustering escalates the PoE or not depends
on the infection probabilities $q_1$ and
$q_2$ and the current PoEs $\mu_1$ and
$\mu_2$. 
\myskip 

\begin{theorem}	\label{thm:1}
Suppose $0 < C_1 < C_2 < 1$ and let
$\bmu^i = (\mu^i_1, \mu^i_2)$, $i = 1, 2$,
be the PoE vector corresponding to 
clustering coefficient $C_i$. Assume that
$\bmu^1 < {\bf 1}$. Then, $\bmu^1 \leq 
\bmu^2$ if
$(1 - \mu^1_1) q_1 \geq (1 - \mu^1_2) q_2$.
Analogously, we have $\bmu^1 \geq \bmu^2$
if $(1 - \mu^1_1) q_1 \leq (1 - \mu^1_2) q_2$.
\end{theorem}
\begin{proof}
A proof is provided in Section~\ref{appen:thm1}. 
\end{proof}

The theorem suggests that the relation between 
clustering coefficient and PoCF is far from simple 
in that it depends on many factors, including the 
degree distribution ${\bf p}_0$ and the
failure probability function $\wp_f$, in a rather
subtle fashion. At the same time, it hints that
their effects can be succinctly summarized by 
the terms in the condition, namely the current 
PoEs $\mu_t$, $t \in \cT$, and infection 
probabilities $q_t$, $t \in \cT$. 

Before we proceed with our
discussion, let us first rewrite the conditions
in the theorem in a slightly different manner:
\beqa
\frac{1 - \mu^1_1}{1 - \mu^1_2} \geq \frac{q_2}{q_1}
\mbox{ and } 
\frac{1 - \mu^1_1}{1 - \mu^1_2} \leq \frac{q_2}{q_1} 
	\label{eq:Equiv}
\eeqa

{\em Remark 1.}
Recall that $1 - \mu_t$ is the PoCF starting with a 
single type $t$ CS. Thus, $(1 - \mu_1) / 
(1 - \mu_2)$ is the ratio of the PoCFs starting
with two different types of CSes.
It is now clear that the 
conditions in the theorem compare the ratio of PoCFs 
for two different types of CSes to that of infection 
probabilities. While one would expect the ratio of 
infection probabilities to play a role, to the best 
of our knowledge, our result is the first to bring 
to light (i) the fact that the ratio of PoCFs, 
$(1 - \mu_1)/(1 - \mu_2)$, plays a similar/important 
role and (ii) a condition that 
tells us when increasing clustering improves 
or degrades the robustness of the system.

{\em Remark 2.}	
A key implication of our finding is the 
following: consider two distinct degree
distributions with the same mean, but one 
degree distribution is more concentrated 
(around the mean) with smaller variance 
than the other. As the 
clustering coefficient becomes larger, the 
fraction of type 2 CSes tends to increase. 
Because a type 2 CS already has two failed 
neighbors compared to a single failed neighbor 
of a type 1 CS, when CS degrees are 
concentrated (around the mean degree), a type
2 CS has fewer neighbors that it can potentially 
infect, thereby producing a smaller number
of children on the average.  

A consequence of this is that it leads to 
higher concentration of failed CSes in a 
local neighborhood. At the same time, 
it hinders spreading failures 
beyond the local neighborhood of already 
failed CSes, especially when the mean degrees 
are not large. Moreover, as we will demonstrate
in the subsequent section, clustering has more 
pronounced impact on PoCFs. 

{\em Remark 3.}
One important fact we should point out is that
Theorem~\ref{thm:1} alone does not
guarantee the monotonicity of PoE (equivalently, 
PoCF) with respect to clustering coefficient.
The reason for this is as follows:
suppose that $0 < C_1 < C_2 < C_3 < 1$ and
$\bmu^i$ is the PoE vector corresponding to 
$C_i$, $i = 1, 2, 3$. Then, it is possible 
to have $\frac{1 - \mu^1_1}{1 - \mu^1_2} 
> \frac{q_2}{q_1}$ while $\frac{1 - \mu^2_1}
{1 - \mu^2_2} < \frac{q_2}{q_1}$ because the
PoEs can change when the clustering coefficient
goes up from $C_1$ to $C_2$. When this
happens, Theorem~\ref{thm:1} tells us
(i) $\bmu^1 \leq \bmu^2$ and $\bmu^1 \leq \bmu^3$
from the first inequality and (ii) $\bmu^2
\geq \bmu^3$ from the second inequality. Thus, 
in principle, we could end up with $\bmu^1 < 
\bmu^3 < \bmu^2$. However, our numerical studies 
provided in the subsequent section indicate
that the monotonicity of PoCF may hold in many 
cases of practical interest.

\section{Numerical Results}
	\label{sec:Numerical}

Our main result in Section~\ref{sec:Main}
tells us that, once we know the PoEs $\mu_t$, 
$t \in \cT$, 
and the infection probabilities $q_t$, 
$t \in \cT$, 
for a fixed clustering coefficient, we can 
determine whether an 
increasing clustering coefficient will 
elevate or lower the resulting PoCF. But,  
this requires the knowledge of, among other
things, the infection probabilities $q_t$, $t 
\in \cT$, which depend on the degree distribution
in the dependence graph and failure probability 
function $\wp_f$. 

The goal of this section is to provide numerical 
results to examine how various parameters 
in our analysis are affected by the clustering
coefficient as we change the degree distribution 
and the failure probability function. To this 
end, we vary the clustering coefficient
over [0.025, 0.40], under two commonly studied 
degree distributions -- power law degree 
distributions and Poisson degree distributions. 

For our study, we consider failure probability
functions of the form $\wp_f(d, n_f) = 
(n_f / d)^\alpha$ with $\alpha \in [0.6, 1.0]$. 
When $\alpha$ is small, even the failure of a single
neighbor causes the failure of a CS with relatively 
high probability. In this sense, $\wp_f(d, n_f)$ 
is less sensitive to $d$. 
On the other hand, for large $\alpha$, 
many neighbors have to fail first before a CS
fails as a result. 
Throughout the section, we assume that
the largest degree of a CS is 20, which we 
denote by $D_{\max}$, and define $\cD := \{1, 2, 
\ldots, D_{\max}\}$.\footnote{We conducted
additional numerical studies with larger
values of $D_{\max}$ and observed similar
qualitative results, which are not reported
here.}

\subsection{Power law degree distributions}

Power law degree distributions assume 
$p_0(d) \propto d^{-\beta}$ for some positive
constant $\beta$. For many real networks, the power 
law exponent $\beta$ lies between 2 and 4 
\cite{Albert00, Lakhina03}. In our study, we
considered $\beta \in [1.5, 4]$. However, 
for $\beta > 2.5$, the resulting PoCF is 
very small and, for this reason, we do not 
present the numbers here.

\begin{figure}[h]
\centerline{
\includegraphics[width=3.45in]{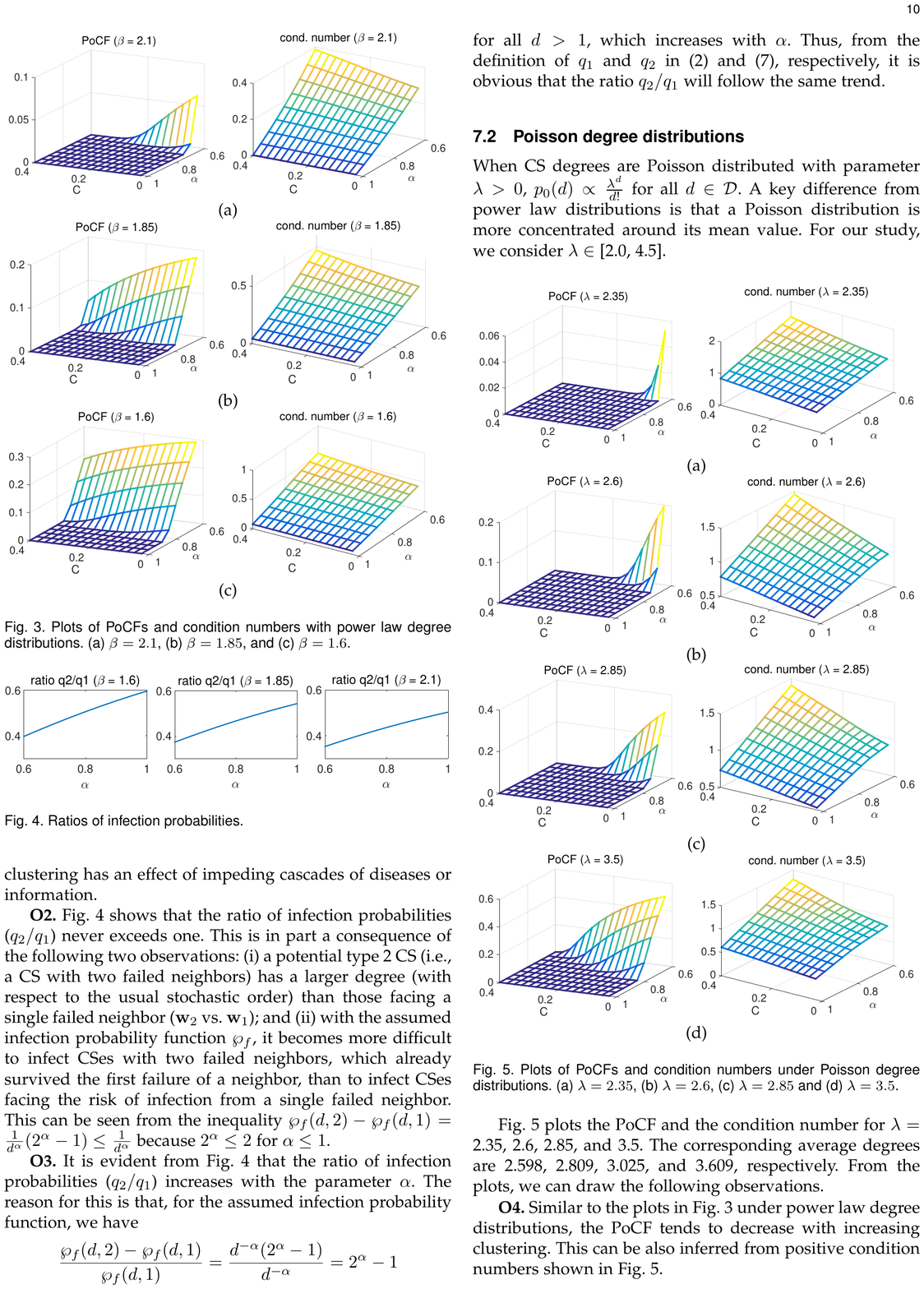}
}
\caption{Plots of PoCFs and condition numbers with power
law degree distributions. 
	(a) $\beta = 2.1$, (b) $\beta = 1.85$, 
and (c) $\beta = 1.6$.}
\label{fig:PoCF-PL}
\end{figure}

\begin{figure}[h]
\centerline{
\includegraphics[width=3.5in]{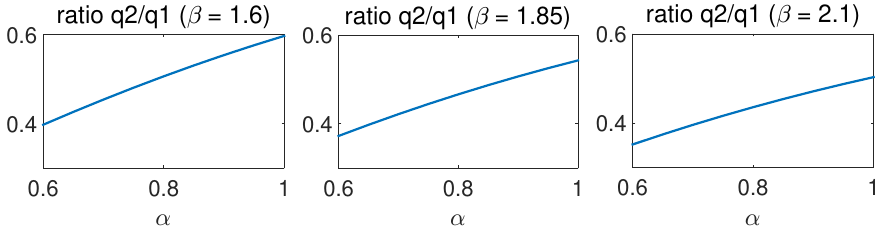}
}
\caption{Ratios of infection probabilities.}
\label{fig:qratio-PL}
\end{figure}

Fig.~\ref{fig:PoCF-PL} plots the PoCF (i.e., 
$1-\mu_0$)
and the condition number as a function of the 
failure probability function parameter $\alpha$ 
and clustering coefficient $C$ for three different
values of power law exponent $\beta$ -- 1.6, 1.85 and 
2.1. Here, the condition number refers to the difference 
$\frac{1-\mu_1}{1-\mu_2} - \frac{q_2}{q_1}$. 
The average degree of CSes under the three values
of $\beta$ is  3.186, 2.548, and 2.089, respectively. 
As mentioned earlier, 
while we examined the scenarios with larger $\beta$, 
the PoCF was too small to be of much interest in 
our opinion.

{\bf O1.} According to Theorem~\ref{thm:1}, when 
the condition number is positive (resp. negative), 
increasing clustering coefficient $C$ reduces 
(resp. elevates) the PoCF. It is clear from 
Fig.~\ref{fig:PoCF-PL} that, for all three values 
of $\beta$, the condition number is 
always positive, and the PoCF decreases with 
clustering coefficient $C$. 
Hence, the plots corroborate our finding in the
theorem. Furthermore, they suggest that
the system is likely to become more robust
against random failures and less prone to 
experience a cascade of failures. 
This observation is consistent
with the findings of \cite{Miller09, 
Zhuang16}, 
which reported that clustering has an effect
of impeding cascades of diseases or information.

{\bf O2.} Fig.~\ref{fig:qratio-PL} shows that the 
ratio of infection probabilities ($q_2/q_1$) never 
exceeds one. This is in part a consequence of 
the following two observations: 
(i) a potential type 2 CS (i.e., a CS with two
failed neighbors) has a larger degree (with
respect to the usual stochastic order) than 
those facing a single failed neighbor ($\bw_2$
vs. $\bw_1$); and (ii) with the assumed
infection probability function $\wp_f$, it 
becomes more difficult to infect 
CSes with two failed neighbors, which already
survived the first failure of a neighbor, than to 
infect CSes facing the risk of infection from 
a single failed neighbor. This can be seen 
from the inequality
$\wp_f(d,2) - \wp_f(d,1) = \frac{1}{d^\alpha}
(2^\alpha - 1) \leq \frac{1}{d^{\alpha}}$ because
$2^\alpha \leq 2$ for $\alpha \leq 1$.

{\bf O3.} It is evident from Fig.~\ref{fig:qratio-PL}
that the ratio of infection probabilities
($q_2 / q_1$) increases with the parameter $\alpha$. 
The reason for this is that, for the assumed infection
probability function, we have
\beqan
\frac{\wp_f(d,2) - \wp_f(d,1)}{\wp_f(d,1)}
\myeq \frac{d^{-\alpha}(2^{\alpha} - 1)}{d^{-\alpha}}
	= 2^{\alpha} - 1
\eeqan
for all $d > 1$, 
which increases with $\alpha$. Thus, from the 
definition of $q_1$ and $q_2$ in (\ref{eq:q1})
and (\ref{eq:q2-2}), respectively, 
it is obvious that the ratio
$q_2/q_1$ will follow the same trend.

\subsection{Poisson degree distributions}

When CS degrees are Poisson distributed
with parameter $\lambda > 0$, $p_0(d)
\propto \frac{\lambda^d}{d!}$
for all $d \in \cD$. A key difference from 
power law distributions is that a Poisson 
distribution is more concentrated around its
mean value. For our study, 
we consider $\lambda \in$ [2.0, 4.5].

\begin{figure}[h]
\centerline{
\includegraphics[width=3.45in]{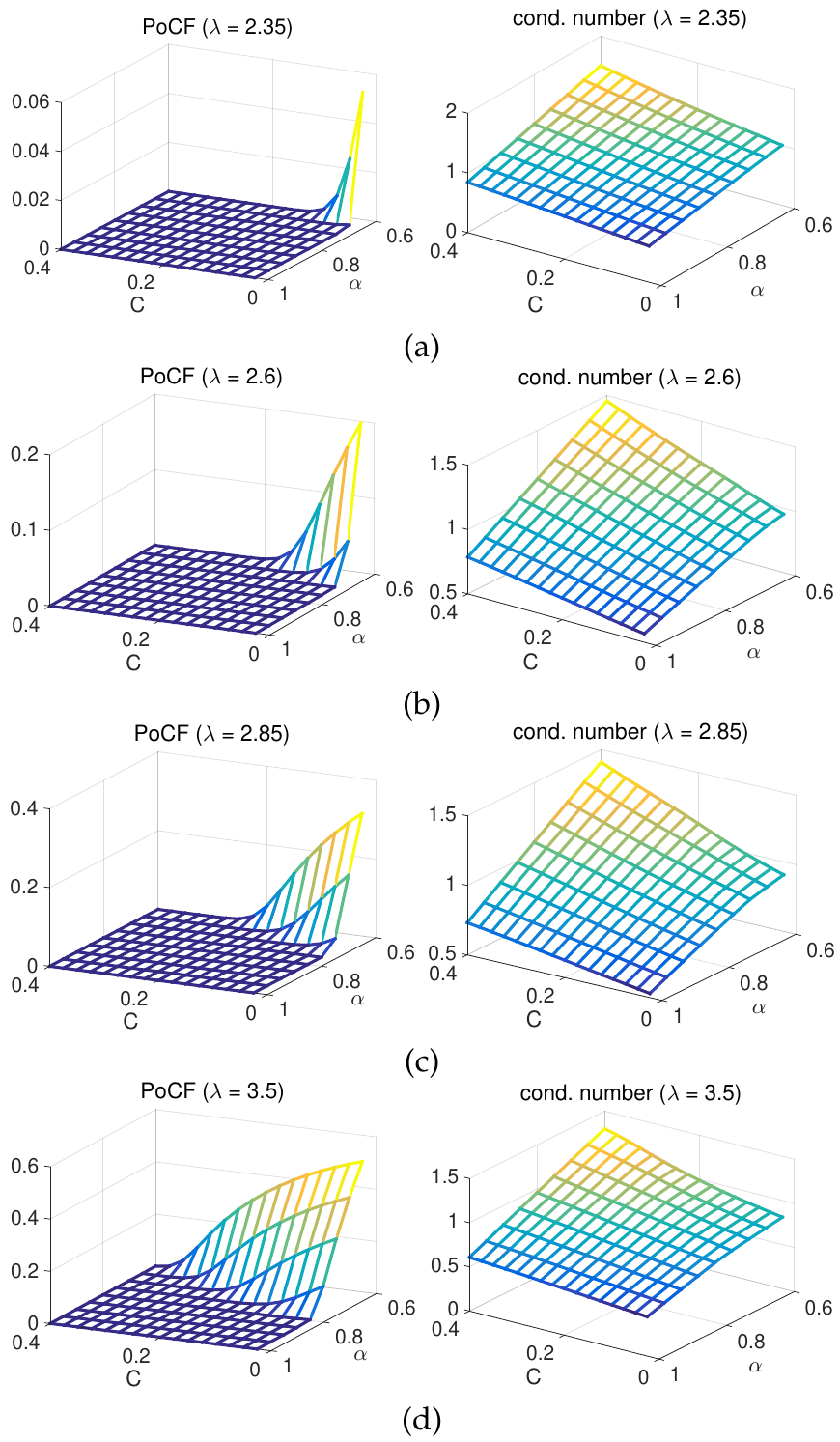}
}
\caption{Plots of PoCFs and condition numbers under 
Poisson degree distributions. 
(a) $\lambda = 2.35$, (b) $\lambda = 2.6$, 
(c) $\lambda = 2.85$ and (d) $\lambda = 3.5$.}
\label{fig:PoCF-Poi}
\end{figure}

\begin{figure}[h]
\centerline{
\includegraphics[width=3.5in]{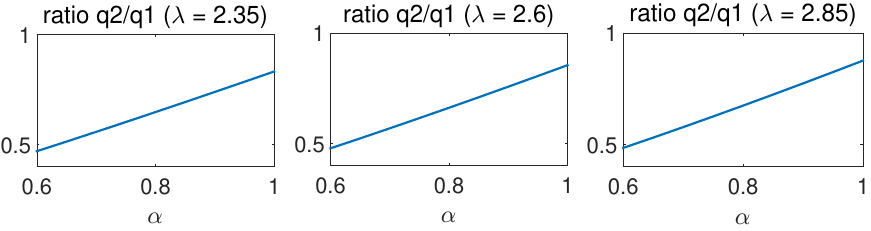}
}
\caption{Ratios of infection probabilities.}
\label{fig:qratio-Poi}
\end{figure}

Fig.~\ref{fig:PoCF-Poi} plots the PoCF and 
the condition number for $\lambda =$ 2.35, 2.6, 
2.85, and
3.5. The corresponding average degrees are 
2.598, 2.809, 3.025, and 3.609, respectively. 
From the plots, we can draw the following
observations.

{\bf O4.} Similar to the plots in Fig.
\ref{fig:PoCF-PL} under power law 
degree distributions, the PoCF tends to decrease with
increasing clustering. 
This can be also inferred from positive
condition numbers shown in Fig.~\ref{fig:PoCF-Poi}.

One noticeable difference is that comparing 
Fig.~\ref{fig:PoCF-Poi} to Fig.~\ref{fig:PoCF-PL}
reveals that the PoCF is much more sensitive and 
decreases more rapidly with 
increasing clustering. This is a 
direct consequence of the earlier observations that 
(i) CS degrees are concentrated around the mean value 
under Poisson degree distributions and (ii) when 
CS degrees are more concentrated around the mean 
value, the clustering tends to impede cascading 
failures and have greater influence on PoCFs, 
especially with small mean degrees
(which lie between 2.598 and 3.609 for scenarios 
shown in Fig.~\ref{fig:PoCF-Poi}). 

\begin{figure}[h]
\centerline{
\includegraphics[width=3.5in]{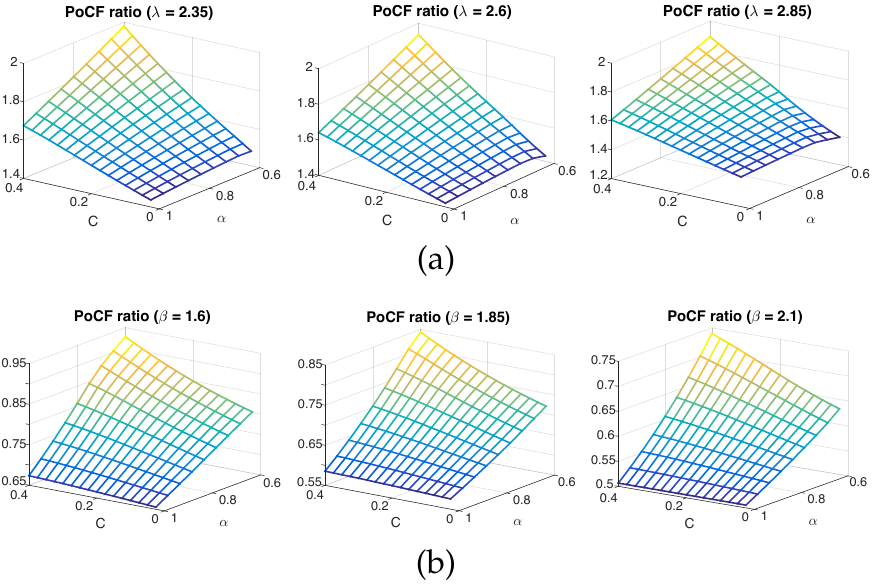}
}
\caption{Ratios of PoCFs $(1-\mu_1) / (1-\mu_2)$. 
(a) Poisson degree distributions for $\lambda =$ 2.35, 2.6, 
	and 2.85, (b) power law degree distribution for $\beta =$ 
	1.6, 1.85, and 2.1.}
\label{fig:PoCF-Ratio}
\end{figure}

{\bf O5.} Comparing the condition numbers
in Figs.~\ref{fig:PoCF-PL} and \ref{fig:PoCF-Poi}
and the infection probability ratios in 
Figs.~\ref{fig:qratio-PL}
and \ref{fig:qratio-Poi}, we can see that 
both the condition numbers and the infection 
probability ratios are considerably larger with 
Poisson degree distributions. This tells us 
that the ratio of PoCFs given by 
$\frac{1-\mu_1}{1 - \mu_2}$, 
which equals the sum of the first two ratios, 
is larger than one with Poisson degree
distributions. This is shown in Fig.
\ref{fig:PoCF-Ratio}(a). Therefore, it corroborates 
Remark 2 in the previous section that, when the 
degree distribution is concentrated, a type 1 
CS is more effective at triggering cascading
failures than a type 2 CS. 

Let $\E{{\bf p}_t}$, $t \in \cT$, be the average
degree among type $t$ CSes. Given that 
$q_2 < q_1$ as shown in Fig.~\ref{fig:qratio-Poi}, 
if $\E{{\bf p}_2} - 
\E{{\bf p}_1} < 1$, clearly a type 1 CS
would serve as a better trigger than a type 2 CS 
because, together, they imply that a type 1 CS on 
the average generates (i) more type 1 children and 
(ii) a larger number of immediate children. 
It turns out that the difference $\E{{\bf p}_2} 
- \E{{\bf p}_1}$ is (considerably) less than one 
in the cases we considered. For example,
for $\alpha = 0.6$ and $\lambda = 2.35$
(resp. $\alpha = 0.8$ and $\lambda = 2.85$), 
we have $\E{{\bf p}_1} = 2.899$ and $\E{{\bf p}_2} 
= 3.303$ (resp. $\E{{\bf p}_1} = 3.194$ and 
$\E{{\bf p}_2} = 3.562$).

Carrying out the same exercise for power law degree 
distributions illustrates that the ratio of PoCFs 
$\big( \frac{1-\mu_1}{1 - \mu_2} \big)$ is less 
than one, as shown in Fig.~\ref{fig:PoCF-Ratio}(b). 
Thus, in this case, the opposite is true, and a 
type 2 CS is more likely to trigger a cascade 
of failures than a type 1 CS. This can be in 
part attributed to the following observation:

Substituting the expressions for the power law
degree distributions and the assumed 
failure probability 
function in the expressions for the conditional
degree distributions ${\bf p}_t$ yields
\beqa
p_t(d) 
& \propto & d^{-(\beta + \alpha - 1)} \
	\mbox{ for all } \ t \leq d \leq D_{\max}.  
		\label{eq:PL-pt} 
\eeqa
From (\ref{eq:PL-pt}), because the conditional 
degree distribution does not decrease quickly 
(e.g., exponentially) with $d$, the difference 
$\E{{\bf p}_2} - \E{{\bf p}_1}$
tends to be larger than one. For instance, 
for $\alpha = 0.8$ and $\beta = 1.6$
(resp. $\alpha = 0.6$ and $\beta = 1.85$), 
we have $\E{p_1} = 3.846$ and $\E{p_2} = 5.941$
(resp. $\E{p_1} = 3.668$ and $\E{p_2} = 5.785$).
Therefore, this suggests that, at least for 
small clustering coefficient $C$, a type 2 CS
is more likely to be successful at setting off
cascading failures than a type 1 CS because
its (expected) degree is much larger.

\section{A proof of Theorem~\ref{thm:1}}
	\label{appen:thm1}

For each $t \in \cT$, define a generating
function $f_t: [0, 1]^{3} \to [0, 1]$, 
where 
\beqa
f_t(s_1, s_2; C) = \sum_{(o_1, o_2) \in \Z_+^{2}} 
	h_t(o_1, o_2; C) \ s_1^{o_1} \ s_2^{o_2}.
	\label{eq:GF}
\eeqa
Then, for a fixed clustering coefficient $C
\in (0, 1)$, the PoE vector $\boldsymbol{\mu}$ 
is given as a fixed point that satisfies 
\beqa
{\bf f}(\bmu; C)
\myeq \bmu \leq {\bf 1}, 
	\label{eq:fixed}
\eeqa
where ${\bf f}(\bs; C) = (f_1(\bs; C), f_2(\bs; C))$.
When $\rho(\bM) > 1$, there exists a unique $\bmu$
that satisfies (\ref{eq:fixed}) with strict
inequality, i.e., $\bmu < {\bf 1}$
\cite{Harris}.

In order to prove the theorem, we will first
prove that, if $(1 - \mu^1_1) q_1 
\geq (1 - \mu^1_2) q_2$ (resp. $(1 - \mu^1_1) q_1 
\geq (1 - \mu^1_2) q_2$), then $\bmu^1 \leq
{\bf f}(\bmu^1; C_2)$ (resp. $\bmu^1 \geq
{\bf f}(\bmu^1; C_2)$).  When this is true, 
Corollary 2~\cite[p. 42]{Harris} tells us
that $\bmu^1 \leq \bmu^2$ (resp. $\bmu^1
\geq \bmu^2$), completing
the proof.

Let us begin with the first part of the theorem. 
Recall from (\ref{eq:ht}) and (\ref{eq:GF}) 
\beqan
f_1(\bs; C)
\myeq \sum_{\bo \in \Z_+^2} h_1(\bo; C)
	\ s_1^{o_1} \ s_2^{o_2} \lb 
\myeq \sum_{\bo \in \Z_+^2} 
	s_1^{o_1} \ s_2^{o_2} 
		\sum_{d \geq o_1 + o_2 + 1} \Big( p_1(d)
	\left[ \frac{d-1}{o_1:o_2} \right] \lb
&& \myhf \times 
	(v^{C}_{1,1})^{o_1} \ (v^{C}_{1,2})^{o_2}
		\ (v^{C}_{1,0})^{d-o_1-o_2-1} \Big). 
\eeqan
After exchanging the order of the two summations, 
we obtain
\beqa
&& \myhb f_1(\bs; C) \lb
\myeq \sum_{d \in \N} p_1(d) \Big( \sum_{\bo \in \Z_+^2:
 	o_1+o_2 \leq d-1}
	s_1^{o_1} \ s_2^{o_2}
	\left[ \frac{d-1}{o_1:o_2} \right] \lb
&& \hspace{0.6in} \times (v^{C}_{1,1})^{o_1} 
	\ (v^{C}_{1,2})^{o_2}
		\ (v^{C}_{1,0})^{d-o_1-o_2-1} \Big) \lb
\myeq \sum_{d \in \N} p_1(d)  \big(
	\underbrace{s_1 \ v^{C}_{1,1}
	+ s_2 \ v^{C}_{1,2} + v^{C}_{1,0}}_\text{
				\large{$=: \psi_1(\bs, C)$}} \big)^{d-1}. 	
	\label{eq:f1C}
\eeqa
The second equality in (\ref{eq:f1C}) follows from 
the well-known equality 
\beqan
\sum_{ \sum_{k=1}^m n_k = N} 
	{{N}\choose{n_1, n_2, \ldots, n_m}} \prod_{k=1}^m
		\left(x_k \right)^{n_k} 
\myeq \left( \sum_{k=1}^m x_k \right)^N, 
\eeqan
where the summation on the left-hand side is over 
non-negative integers $n_k$, $k = 1, 2, \ldots, m$,
whose sum equals $N$. 
Following similar steps, we obtain
\beqa
&& \myhb f_2(\bs; C)
= \sum_{\bo \in \Z_+^2} h_2(\bo; C)
	\ s_1^{o_1} \ s_2^{o_2} \lb 
\myeq \sum_{\bo \in \Z_+^2} 
	s_1^{o_1} \ s_2^{o_2} 
		\sum_{d \geq o_1 + o_2 + 2} \Big( p_2(d)
	\left[ \frac{d-2}{o_1:o_2} \right] \lb
&& \myhf \times 
	(v^{C}_{2,1})^{o_1} \ (v^{C}_{2,2})^{o_2}
		\ (v^{C}_{2,0})^{d-o_1-o_2-2} \Big) \lb
\myeq \sum_{d=2}^\infty p_2(d) 
	\Big( \sum_{\bo \in \Z_+^2:o_1+o_2 \leq d-2}
		s_1^{o_1} \ s_2^{o_2}
	\left[ \frac{d-2}{o_1:o_2} \right] \lb
&& \hspace{0.6in} \times (v^{C}_{2,1})^{o_1} 
	\ (v^{C}_{2,2})^{o_2}
		\ (v^{C}_{2,0})^{d-o_1-o_2-2} \Big) \lb
\myeq \sum_{d=2}^\infty p_2(d)  \big(
	\underbrace{s_1 \ v^{C}_{2,1}
	+ s_2 \ v^{C}_{2,2} + v^{C}_{2,0}}_\text{
		\large{$=: \psi_2(\bs, C)$}} 
			\big)^{d-2}. 	
	\label{eq:f2C}
\eeqa

It is clear from (\ref{eq:f1C}) 
that $f_1(\bmu^1; C_1) = \mu^1_1 \leq f_1(\bmu^1; C_2)$ 
if and only if
$\psi_1(\bmu^1, C_1) \leq \psi_1(\bmu^1, C_2)$. 
Substituting the expressions $v^C_{1,1}
= (1 - C) q_1$, $v^C_{1,2} = C \cdot q_2$ and 
$v^C_{1,0} = 1 - v^C_{1,1} - v^C_{1,2}$ 
from Section~\ref{subsec:ChildDistribution},  
we obtain 
\beqan
\psi_1(\bs, C) 
\myeq s_1(1 - C) q_1 + s_2 C q_2 + 1 - (1-C) q_1 - C q_2.
\eeqan
Hence, after grouping only the terms containing $C$, 
we see that $\psi(\bmu^1, C_1) \leq \psi(\bmu^1, C_2)$ 
if and only if
\beqa
(1 - \mu^1_1) q_1 - (1 - \mu^1_2) q_2 \geq 0.
	\label{eq:appen1}
\eeqa

We now proceed to demonstrate that $\mu^1_2 
= f_2(\bmu^1; C_1) \leq f_2(\bmu^1; C_2)$. From 
(\ref{eq:f2C}), it is obvious that this claim is 
true if and only if $\psi_2(\bmu^1, C_1) \leq 
\psi_2(\bmu^1, C_2)$.  Recall $v^C_{2,1} = (1-C)^2 
q_1$, $v^C_{2,2} = C(2-C) q_2$ and $v^C_{2,0}
= 1 - (1-C)^2 q_1 - C(2-C) q_2$. Plugging in 
these expressions in $\psi_2(\bs; C)$, 
\beqan
\psi_2(\bs, C)
\myeq s_1 (1 - C)^2 q_1 + s_2 C (2-C) q_2 \lb
&& + 1 - (1-C)^2 q_1 - C(2-C) q_2 \lb  
\myeq 1 - (1-C)^2 q_1 (1 - s_1) 
	- C (2-C) q_2 (1 - s_2).  
\eeqan
Collecting only the terms with $C$ in $\psi_2(\bs, C)$,
we get
\beqan
C(2-C) \big( (1-s_1) q_1 - (1-s_2) q_2 \big). 
\eeqan 
Because $C(2-C)$ is strictly increasing over
(0, 1), 
$\psi_2(\bmu^1; C_1) \leq \psi_2(\bmu^1; C_2)$
if and only if
\beqa
(1-\mu^1_1) q_1 - (1-\mu^1_2) q_2 \geq 0,
	\label{eq:appen2} 
\eeqa
which is the same condition in (\ref{eq:appen1})
we obtained earlier.
This completes the proof of the first part of 
the theorem.

The second part of the theorem is a simple 
consequence of the observation that, 
from the proof of the first part, 
$\psi_t(\bmu^1, C_1) \geq \psi_t(\bmu^1, 
C_2)$, $t \in \cT$, if and only if the inequalities 
in (\ref{eq:appen1}) and (\ref{eq:appen2})
go the other way.

\section{Conclusion}
	\label{sec:Conclusion}

We examined the influence of clustering 
in interdependent systems on the 
likelihood of a random failure setting off a 
cascade of failures in a large system. 
We proposed a new model that captures the 
manner in which the triangles alter how a failure 
propagates from a failed system to neighboring 
systems. Utilizing the model, we derived a simple
condition that indicates how increasing clustering
changes the likelihood of experiencing cascading
failures in large systems. This condition also 
hints that, as the degree distribution of the
dependence graph becomes more concentrated, 
higher clustering will help curb the onset of 
widely spread failures by containing them 
to a small neighborhood around an initial
failure.

Our model assumes that the underlying dependence
graph is neutral and exhibits no degree correlations. 
While this helps us isolate the impact of clustering
on the robustness of the system, some real systems 
may display assortative/disassortative mixing. We
are currently working on generalizing the model to
incorporate assortativity, while retaining the separation
of the influence of clustering from that of assortativity. 
In addition, we are in the process of extending the
model to multiplex/multi-layer networks, in order
to investigate the effects of clustering when 
nodes are connected via different types of networks.

\end{document}